\documentclass[11pt]{article}
\usepackage{fullpage}
\usepackage{float}
\usepackage{epsfig}
\usepackage{amsmath}
\usepackage{amssymb}
\usepackage{amsthm}
\usepackage{amstext}
\usepackage{amsmath}
\usepackage{xspace}
\usepackage{latexsym}
\usepackage{verbatim}
\usepackage{url}
\usepackage[ruled,vlined]{algorithm2e}

\usepackage{tikz}
\usetikzlibrary{arrows,positioning,automata}
\usepackage[T1]{fontenc}
\usepackage[utf8]{inputenc}
\usepackage{authblk}

\newtheorem{theorem}{Theorem}

\newtheorem{lemma}{Lemma}

\newtheorem{properti}{Property}

\usepackage{tikz}
\usetikzlibrary{arrows,positioning,automata}
\theoremstyle{definition}
\newtheorem{example}{Example}

\begin{document}
\markboth{G. Askalidis et al.}{Explaining Snapshots of Network Diffusions}

\title{Explaining Snapshots of Network Diffusions: \\Structural and Hardness Results}
\author{Georgios Askalidis}
\author{Randall A. Berry}
\author{Vijay G. Subramanian}
\affil{Northwestern University}
\date{\vspace{-5ex}}
\renewcommand\Authands{ and }

\maketitle              

\begin{abstract}
Much research has been done on studying the diffusion of ideas or technologies on social networks including the \textit{Influence Maximization} problem and many of its variations. Here, we investigate a type of inverse problem. Given a snapshot of the diffusion process, we seek to understand if the snapshot is feasible for a given dynamic, i.e., whether there is a limited number of nodes whose initial adoption can result in the snapshot in finite time. While similar questions have been considered for epidemic dynamics, here, we consider this problem for variations of the deterministic Linear Threshold Model, which is more appropriate for modeling strategic agents. Specifically, we consider both sequential and simultaneous dynamics when deactivations are allowed and when they are not. Even though we show hardness results for all variations we consider, we show that the case of sequential dynamics with deactivations allowed is significantly harder than all others. In contrast, sequential dynamics make the problem trivial on cliques even though it's complexity for simultaneous dynamics is unknown. We complement our hardness results with structural insights that can lead to better understanding of diffusions on social networks under various dynamics.
\end{abstract}
\section{Introduction}\label{intro}
\par Diffusion processes have been widely studied both theoretically and empirically. One of
the main theoretical frameworks is based on modeling a diffusion as the result of a network
game, i.e., a model in which rational agents make decisions to maximize a pay-off that depends on the actions of other agents in way that depends in part on an underlying network structure. The network game that we assume in this paper is one where agents are called upon at each discrete time point to make a rational decision based on 
previous decisions of other neighboring agents. We can think of the action to be concerning the adoption or not of a new technology and assume that all agents start the game with the status quo technology (we will also refer to this state as ``deactivated'' through out the paper). Moreover, we assume that each agent has a non-negative integer threshold which represents the number of her neighbors in the network that need to adopt the new technology (or ``activate'' as we refer to that action throughout the paper) in order for her utility to be maximized by her also choosing to adopt. In an influence maximization setting this model translates to the widely used \textit{Linear Threshold Model}.
 \par The problem that we study in this paper is a generalization of the \textsc{Target Set} problem introduced in \cite{chen2009approximability}, since alongside the network graph $G$, integer budget $k$ and thresholds $t_1, t_2,\ldots, t_n,$ we are also given a subset $S\subseteq V(G)$ that we call a \textit{Snapshot}. We seek to find an initial seed set of size at most $k$ that leads, in finite time, to the activation of \textit{exactly} $S$. For example, this could model a scenario where a snapshot is observed and one seeks to determine the set of nodes that could have started the underlying diffusion. We call this problem the \textsc{Snapshot} problem, and we study four variations of it. When $S=V(G)$, then the \textsc{Snapshot} problem becomes the \textsc{Target Set} problem since we are looking to activate the whole graph. 
 
\par 
We consider two order-dynamics for our network game setting. In the \textit{simultaneous} (or \textit{parallel}) best-response process, at each point in time, all agents best respond to the state of the network simultaneously while in the \textit{sequential} best-response process we chose only one agent to best respond at each point in time.

\par In addition to the linear threshold model, other widely used models for diffusions in social networks are the \textit{Independent Cascade (IC)}, \textit{Susceptible-Infected (SI)}, \textit{Susceptible-Infected-Susceptible (SIS)} and \textit{Susceptible-Infected-Recovered (SIR)} models. We refer to \cite{easley2010networks} for more information on these models.
\par In the same spirit as the SI and SIR models, for each of the two order-dynamics (simultaneous and sequential), we consider two variations of our problem: one that forces the agents we choose in the seed set to commit to remain activated forever and one that allows them to deactivate at a later stage if such an action maximizes their utility. Note that this restriction concerns only the nodes in the initial seed set and that all other nodes always best-respond and so are allowed to deactivate at any point in time in both settings. When we force the seed set to commit to remain activated, the set of activated nodes can only grow (weakly) larger at each time step and so we call this case \textit{monotone}. Hence we get four variations of the \textsc{Snapshot} problem: \textsc{Monotone Simultaneous Snapshot}, \textsc{Simultaneous Snapshot}, \textsc{Monotone Sequential Snapshot} and \textsc{Sequential Snapshot}.
\par In this work, we start by exploring the connections between feasible snapshots under various dynamics and then show that when we are looking for a single initial adopter we can restrict our attention to the closed neighborhood of the given snapshot. Moreover, in the same case, when trying to find an ordering that produces a given snapshots we can ignore all nodes that are not in the snapshot.
Finally we provide various hardness results for all four variations of the problem, most notably that \textsc{Sequential Snapshot} is {\sf NP}-hard even for $k=1$. Finally, we take an interest in the special case of cliques, a graph structure not studied as much in related literature, and show that even though \textsc{Sequential Snapshot} problem becomes easy to solve, the situation is much more complicated under simultaneous dynamics. 
\par As noted previously, one branch of related work is on variations of the \textsc{Influence maximization problem} \cite{kempe2003maximizing}: Given a graph $G$, threshold vector $\vec{t}$ and a budget $k$, choose $k$ nodes to activate, in order to maximize the number of infected nodes. Such problems can be motivated by marketing scenarios where one tries to target specific influential persons by, e.g., giving them some kind of an offer or a free product, with the goal of making the product as popular as possible. Strong hardness and inapproximability results have been shown for even the special case when all agents have threshold 2 \cite{chen2009approximability}, \cite{lu2011approximation}. Other related problems have been studied as well. For example, \cite{agarwal2008identifying}, defines a notion of influence for bloggers in the web and studies the problem of identifying the most influential bloggers. Similarly, in \cite{mathioudakis2009efficient} the authors define the notions of ``starters'' and ``followers'' in social media and try to identify agents from each set. In a different spirit, \cite{netrapalli2012learning} and \cite{rodriguez2011uncovering} study the problem of determining the edges of the network given the activation times of the agents. 
\par The other branch of related work seeks to find the source of a diffusion modeled as arising from a probabilistic epidemic process. Shah and Zaman, \cite{shah2011rumors}, use the SI model and propose a measure they call rumor-centrality to find the single source of a rumor spread. Prakash et al. in \cite{prakash2014efficiently} study the same problem as us but under the probabilistic SI model and they provide experimentally tested heuristics. Similar work, under the IC model in the context of finding users suspected of providing misinformation, has been done in \cite{nguyen2012sources}. Lappas et al. in \cite{lappas2010finding} study the problem of finding the initial set that best explains a given snapshot in a network. For each set of nodes, they define a cost function that represents the difference of the expected final set of activated nodes and the actual observed snapshot, and try to minimize that function. Finally, assuming that information propagates in a social network following the IC model, Gundecha et al. in \cite{gundecha2013seeking} study the problem of finding initial sources as well as other recipients of some information given only a small fraction of the recipients of the information. Even though the problem is {\sf NP}-hard, they provide an efficient heuristic algorithm that they test with real social media datasets. The main difference of our work from this second body of work is the use of the deterministic Linear Threshold Model in contrast to the stochastic IC and SI ones. This can result in significantly different dynamics.

\section{Model}\label{model}
\par We call the general problem we study the \textsc{Snapshot} problem. The input is a tuple $(G, S, \vec{t}, k)$ where $G=(V,E)$ is an undirected network graph, $S\subseteq V(G)$ is a set of nodes we call the \textit{snapshot}, $\vec{t}=(t_1,t_2\ldots,t_n)$, for $n=|V(G)|,$ is a vector of non-negative integer \textit{thresholds}, and $k$ is a positive integer that we call the \textit{budget}. The goal is to find a set $S_0\subseteq V(G)$ of size at most\footnote{Note that the existence of a seed set of size $\leq k$ does not necessarily imply the existence of a seed set of size exactly $k$, since here we take care to activate only $S$ and nothing more.} $k$, that we will call the \textit{initial activated set} or \textit{seed set}, whose activation will, in finite time, cause the activation of \textit{exactly} $S$ for some valid sequence of best responses.\footnote{Note in the case of simultaneous dynamics the sequence of best responses is unique, while for sequential dynamics there are multiple possibilities; we only require that $S$ be activated under one such sequence.} If such $S_0$ exists for $S$, we will call $S$ a \textit{valid} snapshot. Depending the order dynamics used we get the \textsc{Simultaneous Snapshot} and \textsc{Sequential Snapshot} problems. In these versions we don't force the agents in the seed set to commit to remain activated forever and hence they can best respond by deactivating at any time point. When we do force the nodes in $S_0$ to remain activated forever we get the \textit{monotone} version of each of the two problems, which we will call \textsc{Monotone Simultaneous Snapshot} and \textsc{Monotone Sequential Snapshot}, respectively. 
\par Of particular interest will be the case where $k=1$ and hence $S_0=\{u_0\}$. We will then just say that $u_0$ is an \textit{initial adopter} for $S$. It's important to clarify a point here. We do not need the snapshot $S$ to be the \textit{final} state of the activation triggered by $S_0$. Any $S_0$ that in finite time $t$ will produce \textit{exactly} $S$ is considered to be an initial seed set for $S$ \textit{even if at time $t+1$ more nodes will be added to or removed from} $S$.
\begin{example} Suppose that our input graph and thresholds are as shown in Figure \ref{sim}, our budget is $k=2$, and we use monotone simultaneous dynamics. It can be seen that snapshot $S_1=\{u_1,u_2,u_3\}$ is feasible since we can activate $\{u_1,u_3\}$ at time 0, which will activate node $u_2$ at time 1. We don't mind that at time $2$ node $u_4$ will be activated as well. In contrast, snapshot $S_2=\{u_1,u_3,u_4\}$ is not feasible for $k=2$.
 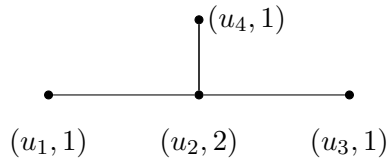
\begin{figure}[h!]
\setlength{\unitlength}{1.2cm}

\begin{center}
\begin{tikzpicture}
\tikzstyle{every node}=[draw,circle,fill=black,minimum size=3pt,
                            inner sep=0.3pt]
                            
\draw (-2,0) node (a) [label=below:${(u_1,1)}$]{};
\draw (0,0) node (b) [label=below:${(u_2,2)}$]{};
\draw (2,0) node (c) [label=below:${(u_3,1)}$]{};
\draw (0,1) node (d) [label=right:${(u_4,1)}$]{};
\path (b) edge (d);
\path (b) edge (c);
\path (b) edge (d);
\path (b) edge (a);

\end{tikzpicture}
\end{center}
\caption{Example of a network and thresholds. The notation $(u_i, t_i)$, that is used throughout the paper, denotes that node $u_i$ has threshold $t_i$.}
\label{sim}
\end{figure}
\end{example}
\par Given that $S=V(G)$, it's implicit from Kempe et al., \cite{kempe2003maximizing}, that finding an influential set of initial adopters of minimum size is {\sf NP}-hard and later Chen, \cite{chen2009approximability}, showed that the problem is hard to approximate within a polylogarithmic factor even when all the thresholds are equal to 2. These will be our starting point towards our hardness results in Section \ref{hardness}.

\section{Structural Results}\label{structural}
In this section we present various results concerning the structure of the snapshots and seed sets under the various dynamics. 
\subsection{Feasible Snapshots of Sequential and Simultaneous Dynamics}
\par We start by understanding the relationship between the feasible snapshots under sequential and simultaneous dynamics. We then show similar results between monotone and non-monotone dynamics.
\begin{lemma}\label{serial} Let $(G,S,\vec{t},k)$ be an instance of the \textsc{Snapshot} problem. If $S$ is feasible for $(G,\vec{t}, k)$ under monotone simultaneous dynamics then it's also feasible for $(G,\vec{t}, k)$ under monotone sequential dynamics.
\end{lemma}
\begin{proof}
Suppose that $S$ is feasible under simultaneous dynamics and $S_0$ is a seed set for $S$. We can then create an ordering that produces $S$, also starting from $S_0$, under sequential dynamics: just take the order with which the nodes were activated. Break the ties between nodes that were activated in the same time period arbitrarily. It can be seen that this indeed produces $S$ under sequential dynamics.
\end{proof}

As the following example shows, the reverse is not true. There are snapshots that are feasible under sequential dynamics (monotone or non-monotone) that cannot be created under simultaneous (monotone or non-monotone) dynamics.
\begin{example} \label{simex}
Consider the graph and thresholds shown in Figure \ref{sim} and assume we have $k=1$, i.e., we are allowed only one initial adopter.

                            

It can be seen that the snapshot $S_1=\{u_2, u_3\}$ is feasible under monotone or non-monotone sequential dynamics since we can activate $u_2$ at time 0 and at time 1 choose $u_3$ to best respond but $S_1$ is not feasible under monotone nor non-monotone simultaneous dynamics.
\end{example}


\par Even though the above example shows that the set of feasible snapshots for sequential and simultaneous dynamics are not the same, the next lemma shows that when we are looking to activate the whole graph, i.e., when $S=V(G)$, and there are no deactivations allowed, then the dynamics are indeed equivalent. 

\begin{lemma}\label{serial2} Let $(G, \vec{t}, S=V(G),k)$ be an instance of the \textsc{Snapshot} problem.
Then $S$ is feasible under monotone sequential dynamics if and only if it is feasible under monotone simultaneous dynamics.
\end{lemma}
\begin{proof} If $S$ is feasible under simultaneous dynamics, then from Lemma \ref{serial} it is feasible under sequential dynamics as well. 
Suppose now that $S$ is feasible under sequential dynamics and that the ordering $u_1,u_2,\ldots,u_n$ produces $S$. It is enough to notice that the activation time of node $u_i$ starting from the same seed set can only be earlier under simultaneous dynamics (since it gets activated right when it has the appropriate number of activated neighbors). That means that at every time point $i$ either the set of nodes activated at that time, $S_i$, contains $u_i$ or $u_i$ is already activated in a previous time step. Since there is a time point for every $u_i$ we know that in finite time all nodes will be activated.
\end{proof}

The key difference that makes Lemma \ref{serial2} work for $S=V(G)$ but not in general is that we don't have the issue of {\it over-activating}, i.e., activating more nodes than are in the snapshot. As can be seen by chosing $u_2$ as an initial adopter in Example \ref{simex}, this can occur when $S$ is a strict subset of $V(G)$. Under sequential dynamics we can ensure that a node is not activated by simply never selecting it to best respond, while with simultaneous dynamics we do not have that freedom. 

\par No containment relation holds between the sets of feasible snapshots under simultaneous and sequential dynamics when we don't require monotonicity, as shown in the next example.

\begin{example} Assume we have the same graph and thresholds as shown in Figure \ref{sim} and $k=1$. As shown in Example \ref{simex}, $S_1=\{u_2,u_3\}$ is feasible under sequential dynamics but not under simultaneous.
\par In the same figure, it can be seen that $S_4=\{u_1,u_3, u_4\}$ is feasible under non-monotone simultaneous dynamics since we can choose $u_2$ as our seed. In the next round, $u_1$, $u_3$ and $u_4$ will activate since they had the appropriate number of activated neighbors in the previous round, and $u_2$ will deactivate since it had 0 activated neighbors in the previous round. In contrast, it can be seen that there is no seed set of size $1$ that can produce $S_4$ under sequential dynamics (monotone or non-monotone).
\end{example} 
\par As Lemma \ref{feasible_sim_2} shows, if $S$ is feasible under monotone sequential dynamics then it's also feasible under non-monotone sequential dynamics. The reverse is not true though, as shown in Example \ref{example_seq_mono_}. No containment relation holds between the feasible snapshots of monotone and non-monotone simultaneous dynamics, as shown in Examples \ref{example_sim_non_sim} and \ref{example_non_sim_sim}.

\begin{lemma}\label{feasible_sim_2}
Let $(G,S,\vec{t},k)$ be an instance of the \textsc{Snapshot} problem. Then if $S$ is feasible under monotone sequential dynamics, it is also feasible under sequential dynamics.
\end{lemma}
\begin{proof}Let $S_0$ and $\mathcal{O}$ be the seed set and ordering that produces $S$ under monotone sequential dynamics. Then by taking $S_0$ and removing all occurrences of $S_0$ (except the very first one) from $\mathcal{O}$, we can produce $S$ under non-monotone sequential dynamics as well, since $S_0$ will never be deactivated.
\end{proof}
The reverse is not true, as shown in the following example.
\begin{example}\label{example_seq_mono_}
Suppose we have the network and thresholds as shown in Figure \ref{2_6_3} with $k=2$ and $S=\{u_9, u_{10}, u_{11}\}$. Then $S$ is feasible under non-monotone sequential dynamics using the starting set $\{u_1, u_2\}$, but $S$ is not feasible under monotone sequential dynamics.
\end{example}
\begin{figure}\label{counter_2}
\begin{center}
\begin{tikzpicture}
\tikzstyle{every node}=[draw,circle,fill=black,minimum size=3pt,
                            inner sep=0.3pt]
                            
\draw (-1,1) node (u1) [label=left: ${(u_1,7)}$]{};
\draw (-1,0) node (u2) [label=left: ${(u_2,7)}$]{};                            
\draw (1,3) node (u3) [label=right: ${(u_3,2)}$]{};
\draw (1,2) node (u4) [label=right: ${(u_4,2)}$]{};
\draw (1,1) node (u5) [label=right: ${(u_5,2)}$]{};
\draw (1,0) node (u6) [label=right: ${(u_6,2)}$]{};
\draw (1,-1) node (u7) [label=right: ${(u_7,2)}$]{};
\draw (1,-2) node (u8) [label=right: ${(u_8,2)}$]{};
\draw (2,2.5) node (u9) [label=right: ${(u_9,3)}$]{};
\draw (2,0.5) node (u10) [label=right: ${(u_{10},3)}$]{};
\draw (2,-1.5) node (u11) [label=right: ${(u_{11},3)}$]{};

\path (u1) edge (u3);
\path (u1) edge (u4);
\path (u1) edge (u5);
\path (u1) edge (u6);
\path (u1) edge (u7);
\path (u1) edge (u8);

\path (u2) edge (u3);
\path (u2) edge (u4);
\path (u2) edge (u5);
\path (u2) edge (u6);
\path (u2) edge (u7);
\path (u2) edge (u8);

\path (u3) edge (u9);
\path (u4) edge (u9);
\path (u5) edge (u10);
\path (u6) edge (u10);
\path (u7) edge (u11);
\path (u8) edge (u11);

\end{tikzpicture}
\end{center}
\caption{When $k=2$, the only way to activate $S=\{u_9,u_{10},u_{11}\}$ is by taking the seed set being $S_0=\{u_1,u_2\}$}
\label{2_6_3}
\end{figure}
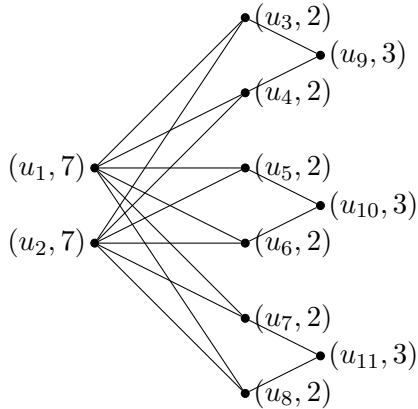

No containment relation holds between the feasible snapshots of monotone and non-monotone simultaneous dynamics, as shown in the next two examples.

\begin{example}\label{example_sim_non_sim}
Consider the graph in Figure \ref{sim} and $k=1$. Then $S=\{u_1,u_3, u_4\}$ is feasible under non-monotone simultaneous dynamics by taking $u_2$ as the initial adopter and letting the process run for one step, but it is not feasible under monotone simultaneous dynamics.
\end{example}

\begin{example}\label{example_non_sim_sim}
Consider the graph in Figure \ref{sim} and $k=1$. Then $S_1=\{u_1, u_2, u_3, u_4\}$ is feasible under monotone simultaneous dynamics with starting set $u_2$, but it is not feasible under non-monotone simultaneous dynamics.
\end{example}

\subsection{Distance Between the Seed Set and the Snapshot}
We study next the distance that a seed set can have from an observed snapshot. Clearly, when we allow for no deactivations the seed set must be part of the observed snapshot and hence the distance is zero. When we have non-monotone sequential dynamics, we show below that for the case of $k=1$, the seed set can have distance at most 1 from the observed snapshot. We then show that this is not true when we have $k\geq 2$ or non-monotone simultaneous dynamics.

\begin{lemma}\label{neighbor}
Let $I=(G, S, \vec{t},1)$ be an instance of the \textsc{Snapshot} problem. If $S$ is valid under sequential dynamics, then there exists an initial adopter $u_0$ for $S$ in $N[S]$.
\end{lemma}
\begin{proof} Since all nodes but $u_0$ are best responding, if $u_0$ never gets deactivated, then the lemma is trivially true since $u_0\in S$. So we assume that $u_0$ is not in $S$. We then have two cases: either some vertex $u\in S$ is connected with $u_0$ or for all vertices $u\in S$ we have $(u,u_0)\not\in E(G)$. In the first case, the lemma is, again, trivially true so we concentrate on the second. Pick any $u\in S$. We will argue that there is a set $A\subseteq V(G)$ with $dist(A,u_0)\leq 1$, that remains activated forever even if $u_0$ gets deactivated and hence $A\subseteq S$. This will give us that $dist(S,u_0)\leq 1$.
\par Let's call $S_i$ the set of nodes that were activated at time $i$. Hence $S_0=\{u_0\}$. We construct $A$ by starting from an arbitrary $u\in S$ and then follow the ``activation path'' of $u$ back to $u_0$: At the $(j+1)$-th iteration, for all $v$ that were added to $A$ in the $j$-th iteration we add to $A$ all nodes in $N(v)$ that were activated by best responding \textit{before} $v$. If no such nodes exists, it means that we have reached $S_1$, the nodes that were activated at time 1 and hence are neighbors of $u_0$. Since at each iteration we add to $A$ nodes with strictly decreasing activation times we are guaranteed to intersect with $S_1$ in finite time and at that point we will have that $dist(A,u_0)=1$. Note that this process could take us outside of $S$ but we will show that it won't. Moreover this is a theoretical construction and hence we don't worry about actually finding these nodes. We will argue that even when $u_0$ gets deactivated, none of the vertices of $A$ can be deactivated and hence $A\subseteq S$, which will  prove that $dist(S,u_0)\leq 1$, contradicting that there is no $u\in S$ adjacent to $u_0$.
\par Assume that there is a vertex in $A$ that gets deactivated and take $u_1$ to be the first such. We have two cases: a) $u_1\in S$ and b) $u_1\notin S$. In the first case, and since we are under the assumption that no node in $S$ shares an edge with $u_0$ we have that in order for $u_1$ to be deactivated, at least one of it's neighbors that was activated before it must be deactivated first. Since we have included all such neighbors in $A$ as well, this contradicts the first deactivator definition of $u_1$. Suppose now that $u_1\notin S$. Let's call $S_i$ the set of nodes in $A\setminus S$ that were activated at time $i$. Notice that, by construction, each node in $S_i$ has at least a neighbor that was activated at time $i+1$ (that neighbor could in $S$). Let $u_1\in S_i$ for some $i$. Notice that when $u_1$ best responds by activating in time $i$, $u_0$ contributes at most 1 to it's activation and none of it's neighbors that were activated in time $i+1$, are yet activated. Once these neighbors do get activated they increase the number of activated neighbors of $u_1$ by at least one. Therefore, even if $u_0$ gets deactivated $u_1$ still has at least $t_{u_1}$ activated neighbors and cannot be the first node to be deactivated.
\end{proof}

\par The main idea behind the proof is that we follow the activation path from a node in $S$ back to $u_0$ and then argue that this path cannot be deactivated because every node $u$ in that path has at least one `down-stream' neighbor $v$ that it is responsible for activating, i.e., $v$ was activated after $u$ and it's activation required $u$ to be active. Therefore, even if $u_0$ deactivates, $v$ will ensure that $u$ still has the appropriate number of activated neighbors. Examples  \ref{seq_nesse} and \ref{k_1_nesse} show that the assumptions of sequential dynamics and $k=1$ respectively, are necessary for Lemma \ref{neighbor} to hold.

\begin{example} \label{seq_nesse}This example shows that if instead of sequential with deactivations dynamics we had simultaneous with deactivations dynamics, then Lemma \ref{neighbor} does not hold for any $k$.
Suppose we have $k=1$ and thresholds and graph as in Figure \ref{diamond_2}. Then if the input snapshot is $S=\{u_4, u_7\}$. Then the only seed set of size $1$ for $S$ is $\{u_1\}$. If we activate $u_1$ at time 0, it will deactivate at time 1 and never get activated again. $u_2, u_3, u_5, u_6$ will get activated at time 1 and deactivate at time 2. $u_4$ and $u_7$ will get activated at time $2$. At that point they will have distance 2 from $u_1$. Notice that by copying the graph of Figure \ref{neighbor} $k$ times we can adjust it to work as a counter example for Lemma \ref{neighbor} even for $k\geq 2$. We can make the distance of the snapshot from the seed set larger too, by adding more intermediate nodes that will be activated and deactivated at the next round. 
\end{example}

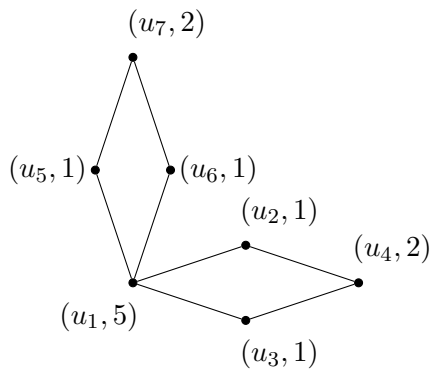
\begin{figure}
\begin{center}
\begin{tikzpicture}
\tikzstyle{every node}=[draw,circle,fill=black,minimum size=3pt,
                            inner sep=0.3pt]
                            
\draw (-2,0.5) node (u1) [label=below left: ${(u_1,5)}$]{};                            
\draw (-0.5,1) node (u2) [label= above right:${(u_2,1)}$]{};
\draw (-0.5,0) node (u3) [label=below right:${(u_3,1)}$]{};
\draw (1,0.5) node (u4) [label=above right: ${(u_4,2)}$]{};
\path (u1) edge (u2);
\path (u1) edge (u3);
\path (u2) edge (u4);
\path (u3) edge (u4);

\draw (-2.5,2) node (u5) [label=left:${(u_5,1)}$]{};
\draw (-1.5,2) node (u6) [label=right:${(u_6,1)}$]{};
\draw (-2,3.5) node (u7) [label=above right:${(u_7,2)}$]{};
\path (u1) edge (u5);
\path (u5) edge (u7);
\path (u7) edge (u6);
\path (u1) edge (u6);

\end{tikzpicture}
\end{center}
\caption{Counter-example for Lemma \ref{neighbor} for simultaneous with deactivation dynamics when $k\geq 1$}
\label{diamond_2}
\end{figure}

\begin{example} \label{k_1_nesse}
Here we give an example where even with serial dynamics Lemma \ref{neighbor} does not extend to the general case of $k\geq 2$. Consider the graph and thresholds as shown in Figure \ref{2_6_3} and suppose that $S=\{u_9, u_{10}, u_{11}\}$ and $k=2$. The only feasible solution is if we take the seed set to be $\{u_1, u_2\}$, which has distance two from the snapshot.
\end{example}

\subsection{The Clearing Lemma}
Lemma \ref{neighbor} states that when we are looking for a single initial adopter, we can just look in the neighborhood of the snapshot $S$. We show next, that we can reduce the search space further by simply ignoring all nodes except $u_0$ that are not in $S$. This shows that in the case of sequential dynamics and $k=1$, the activation cannot 'pass through', i.e., use a node $v$ to activate another node $u\in S$ and then leave $v$ deactivated. Again, the result is trivially true when we have monotone dynamics and hence we concentrate on the non-monotone case.

\begin{lemma}\label{clearlemma}
Let $I=(G,\vec{t},S,1)$ be an instance of \textsc{Sequential Snapshot} and $u_0\in V(G)$. Then there is an ordering of $V(G)$ that produces $S$ with $u_0$ as the initial adopter if and only if there is an ordering of $S\cup u_0$, that produces $S$ with $u_0$ as the initial adopter. Moreover, no node in that ordering other than (possibly) $u_0$ ever gets deactivated.
\end{lemma}
\begin{proof}
The ``only if'' direction is trivial: if there is an ordering of $S\cup u_0$ that produces $S$, then there is an ordering of $V(G)$ that produces $S$. \\
Suppose now that $S$ is feasible and $u_0$ is an initial adopter. There exists an ordering that produces $S$, $\mathcal{O}=(u_0, u_{i_1}, \ldots, u_{i_k})$, of $V(G)$. We will argue that we can ignore all nodes in the ordering $\cal{O}$ except nodes in $S$ and $u_0$. We notice that there are three types of vertices in $\cal{O}$: (i) nodes that never activate, (ii) nodes that activated and remained that way and (iii) nodes that activated and then deactivated. We can ignore the type (i) vertices since they never activated so it's like they never were selected to play. The type (ii) vertices are precisely $S$. The interesting case is the case of type (iii) vertices so let $D$ be the set of those nodes. Let's assume for now that every node can appear in the ordering at most twice (so a node can deactivate at most once). This will help us make the point of the proof more clear. We will lift that restriction directly after. We first notice that if $u_0\notin D$ then $D=\emptyset$, since every node, other than $u_0$, activated by best responding. If $D\neq \emptyset$, we sort the nodes in an increasing order of their deactivation time and it's clear that $u_0$ is the first node in that sorted order. Let $u_1$ be the next node in that order and let $i$ and $j$ be it's activation and deactivation time respectively. Since at time $i$, $u_1$ best responded by activating it means that there were at least $t_{u_1}$ activated nodes in $u_1$'s neighborhood. Similarly, since $u_1$ best responded by deactivating at time $j$, it means that there were at most $t_{u_1}-1$ activated nodes in $N(u_1)$. Since $u_0$ was the only node that got deactivated before time $j$ (by the way we chose $u_1$) and that deactivation was enough to bring the number of activated neighbors of $u_1$ from at least $t_{u_1}$ to at most $t_{u_1}-1$, it means that at time $i$ there were \textit{exactly} $t_{u_1}$ active neighbors of $u_1$ and in the time interval $i$ to $j$ no other neighbor of $u_1$ was activated. Hence, we can ignore $u_1$ from the ordering without affecting any other node. Continuing the same argument to the set $D\setminus u_1$ we conclude we can ignore all nodes in $D$ but $u_0$. 
For the general case when nodes can be included in the ordering an arbitrary number of times and hence they can be deactivated multiple times we use a similar argument. Again, we take $D$ to be the set of nodes that were deactivated at least once and sort them in increasing order of their \textit{first} deactivation time. Then $u_0$ is still the first node of that ordering and we take $u_1$ to be the immediately next node. We use the same argument as above to conclude that we can ignore the first two occurrences of $u_1$ in the ordering and repeat until no occurrences of $D\setminus u_0$ remain
\end{proof}

\par As Example \ref{counter_example_cleari_lemma} shows, Lemma \ref{clearlemma} does not extend to the case where $k\geq 2$.

\begin{example} \label{counter_example_cleari_lemma}
Consider again the graph in Figure \ref{counter_2}. Even though nodes like $u_3$ and $u_4$ are not part of the snapshot $S=\{u_9,u_{10}, u_{11}\}$ nor of the seed set $\{u_1, u_2\}$, they are part of the only ordering that produces it. 
\end{example}

\section{Hardness}\label{hardness}
In this section we study the computational complexity of the various versions of the \textsc{Snapshot} problem discussed in this paper.
This work extends the already rich literature on the hardness of the Influence Maximization problem, which was first formulated and proved to be {\sf NP}-hard by Kempe et al. in \cite{kempe2003maximizing}, and \textsc{Target Set} (which is the minimization variant of the Influence Maximization problem) that was proved to be {\sf APX}-hard even in cases of restricted threshold values by Chen in \cite{chen2009approximability}. The decision versions of these two problems coincide and hence we have that \textsc{Target Set} is {\sf NP}-hard. Some tractable cases have also been shown for the \textsc{Target Set} problem. Chen in \cite{chen2009approximability} gives a linear-time algorithm for trees and Ben-Zwi et al. in \cite{ben2009exact} generalized the result by solving the problem in graphs of constant treewidth. Moreover, even though for any constant $k$ the \textsc{Target Set} problem can be solved in $O(n^{k+1})$ time, the problem is {\sf W[2]}-hard for undirected graphs \cite{nichterlein2010tractable} and {\sf W[P]}-hard for directed graphs \cite{eickmeyer2008approximation}. For more on parameterized complexity we refer to \cite{niedermeier2006invitation}.

\par We use the {\sf NP}-hardness of the \textsc{Target Set} problem, \cite{kempe2003maximizing}, \cite{chen2009approximability}, as our starting point for proving the following theorem.

\begin{theorem}\label{theorem1}
\textsc{Simultaneous Snapshot}, \textsc{Monotone Simultaneous Snapshot} and \textsc{Monotone Sequential Snapshot} are all {\sf NP}-hard, even for the case that all thresholds are less than or equal to 2.
\end{theorem}

\par The result follows from Lemmas \ref{hard_lemma_1}, \ref{hard_lemma_2} and \ref{hard_lemma_3} below that show the individual {\sf NP}-hardness results for each of the three problems. 

\begin{lemma}\label{hard_lemma_1}
The \textsc{Monotone Simultaneous Snapshot} is {\sf NP}-hard.
\end{lemma}
\begin{proof} When $S=V(G)$ the problem corresponds to finding a seed set of size at most $k$ that activates the whole graph. This is exactly the {\sf NP}-hard \textsc{Target Set} problem. The results follows.
\end{proof}

\par In Lemma \ref{serial2} we showed that when the snapshot in the input is the whole vertex set of the graph, then it's feasible under Simultaneous Without Deactivations dynamics if and only if it's feasible under Sequential Without Deactivations dynamics. We use that lemma to extend the {\sf NP}-hardness result of Lemma \ref{hard_lemma_1} to the Sequential Without Deactivations case.

\begin{lemma} \label{hard_lemma_2}
The \textsc{Monotone Sequential Snapshot} is {\sf NP}-hard.
\end{lemma}
\begin{proof}
We take the special case of  the \textsc{Monotone Simultaneous Snapshot} problem where $S=V(G)$ for which we know it's {\sf NP}-hard from Lemma \ref{hard_lemma_1}. Using Lemma \ref{serial2} we know that snapshot $S=V(G)$ will be feasible for graph $G$ under sequential without deactivations dynamics if and only if it's feasible under simultaneous dynamics. The result follows.
\end{proof}

\begin{lemma}\label{hard_lemma_3}
The \textsc{Simultaneous Snapshot} problem is {\sf NP}-hard.
\end{lemma}
\begin{proof}
We reduce from the \textsc{Monotone Simultaneous Snapshot} in such a way that even though deactivations are allowed they don't hurt the activation process. Given an input $(G,S,\vec{t},k)$ for the \textsc{Monotone Simultaneous Snapshot} problem, we will create a new instance $(G',S,\vec{t}',k)$ for the \textsc{Simultaneous Snapshot} such that the latter has a solution if and only if the former one has too.
\par We create $G'$ by creating for every node in $v\in V(G)$, $t_v$ dummy nodes that we connect only with $v$ and we assign them a threshold of one. Additionally, we create another dummy node also with threshold 1 and connect it to $v$ as well as the previous $t_v$ dummy neighbors of $v$. Hence $v$ will have $t_v+1$ dummy neighbors one of whom connected to the other $t_v$ as well. This creates $G'$ and $\vec{t}'$. We leave $S$ and $k$ same.

$(\Leftarrow)$ Suppose that $(G',S,\vec{t}',k)$ has a solution under non-monotone simultaneous dynamics. It can be seen that if there is a node $v\in V(G')\setminus V(G)$ in the seed set for $S$, we can replace it by it's neighbor in $V(G)$ and get the same result. Then that seed set will also be a seed set for $S$ in $(G,\vec{t})$ under monotone dynamics too.

$(\Rightarrow)$ Suppose now that $(G,S,\vec{t},k)$ has a solution under monotone simultaneous dynamics. The new dummy nodes that we added in $G'$ will ensure that once a node $v\in V(G)$ gets activated, in the next round, it will have at least $t_v+1$ activated neighbors, of which $t_v+1$ of them will not be deactivated ever again. Hence even if $v$ got deactivated for one round, it gets reactivated again and stays activated for ever after that. Hence if we have a seed set for $S$ in $(G,\vec{t})$ it will also be a seed set for $(G',\vec{t}')$ as well.
\end{proof}

\par All three reductions in Lemmas \ref{hard_lemma_1}, \ref{hard_lemma_2} and \ref{hard_lemma_3} are \textit{parameter preserving} and as such they carry over the {\sf W[2]}-hardness result shown in \cite{nichterlein2010tractable} for the \textsc{Target Set}. Hence we get the following theorem as well.
\begin{theorem}\label{theorem2}
\textsc{Simultaneous Snapshot}, \textsc{Monotone Simultaneous Snapshot} and \textsc{Monotone Sequential Snapshot} are all {\sf W[2]}-hard when parameterized by the size of the solution, $k$.
\end{theorem}

\par Finally, all three reductions in Lemmas \ref{hard_lemma_1}, \ref{hard_lemma_2} and \ref{hard_lemma_3} are also \textit{approximation preserving}, \cite{crescenzi1997short}, and as such they carry over the approximation hardness shown by Chen, \cite{chen2009approximability}, for the \textsc{Target Set}.

\begin{theorem}
The optimization version of \textsc{Simultaneous Snapshot}, \textsc{Monotone Simultaneous Snapshot}, \textsc{Sequential Snapshot} and \textsc{Monotone Sequential Snapshot} even for the case when all thresholds are less than or equal to 2, cannot be approximated within the ratio of $O(2^{\log^{1-\epsilon}n})$, for any fixed constant $\epsilon>0$, unless {\sf NP}$\subseteq${\sf DTIME}$(n^{polylog(n)})$.
\end{theorem}

\par Nevertheless, for constant $k$, all three problems discussed so far are solvable in time $O(n^{k+1})$ by a brute force search, and hence are polynomial time solvable for $k=1$. We show next that when we have non-monotone sequential dynamics the problem becomes {\sf NP}-hard even for $k=1$.

\begin{theorem}
The \textsc{Sequential Snapshot} problem is {\sf NP}-hard even for $k=1$.
\end{theorem}
\begin{proof} Because we are in the case where $k=1$, we can assume that a potential seed node $u_0$ is given along the input and the decision problem becomes finding out if $S$ is feasible with $u_0$ as the initial seed node. This version of the problem is polynomial time equivalent to the original version since if we could solve this in polynomial time, then we could solve the original problem by just taking every $u\in N[S]$ as a potential $u_0$. This would increase the complexity of the algorithm only by a linear factor. Given an instance $I=(G, \vec{t}, k)$ of the \textsc{Target Set} problem, we define a new instance $I'=(G', S\subseteq V(G'),\vec{t'}, k')$ of the \textsc{Sequential Snapshot} problem as follows: for every node $u_i\in V(G),$ with threshold $t_{i}$ we define $t_{i}$ new nodes $v_i^1,\ldots, v_i^{t_i}\in V(G')$ and we connect these new nodes only to $u_i$. We then add, for every node $u_i\in V(G)$ a new node $u'_i\in V(G')$ and connect $u'_i$ to $v_i^1,\ldots, v_i^{t_i}$. We set the threshold of all new nodes to be equal to 1. Finally, we add a new node $v_0\in V(G')$ and set its threshold to be equal to $k+1$. We connect $v_0$ to all $u'_i$. We set $S=V(G')\setminus u_0$\\
$(\Rightarrow)$ Suppose that there exists a target set $A=\{u_1,u_2,\ldots,u_k\}$ of size $k$ in $V(G)$. We can then find an ordering in $V(G')$ that activates $S$: $u_0,u'_1,\ldots, u'_k, u_0,\ldots$ First of all, since $u_0$ has threshold $k+1$, it will deactivate when it's asked to best respond (i.e. the second time it appears in the ordering). Then notice that the activation of $u'_i$ is enough to activate $u_i$. This is because once $u'_i$ gets activated it's neighboring $v_i^1, v_i^2,\ldots, v_i^{t_i}$ nodes of threshold 1 are ready to get activated as well. Since they are $i$ of them, once they get activated $u_i$ has enough activated neighbors to get activated as well. So if we activate $u'_1,\ldots, u'_k$ we can obtain an ordering that activates $u_1,\ldots, u_k$ as well. Since these nodes are a starting set of $V(G)$ we can find an ordering that activates all $V(G)$. Hence if we know a Target Set for $V(G)$ we can find an ordering that activates all of $S=V(G')\setminus u_0$ with $u_0$ as an initial adopter.\\
$(\Leftarrow)$ Suppose that there exists an ordering of $V(G')$ that activates $S=V(G')\setminus v_0$. We have two cases here, depending if $v_0$ is the first node in that ordering or not. Suppose it's not $u_0$, then we by the clearing lemma we can ignore $u_0$. Notice that we can assume that the initial node is in $V(G)$ since, if it's not and instead it's some $u_i'$ we can replace it with $u_i$ and get at least as many nodes as before. Hence in this case we have a node $u_i$ that activates the whole $V(G)$, so it's a starting set of size 1. Since $k\geq 1$, there is a starting set of size at most $k$ for $I$.
Suppose now that the initial node is $u_0$. Since $u_0$ is not in the snapshot, it means that at some time point it best-responded by deactivating and since it has a threshold of $k+1$, at most $k$ of it's neighbors, $\{u_1', u_2',\ldots, u_{\ell}'\}$ for some $\ell\leq k$,  were activated at that time point. By taking $A=\{u_1,\ldots, u_{\ell}\}$ we have a starting set for $G$. That's because $u_i'$ can only activate $u_i$, and since the activation of $\{u_1',\ldots, u_{\ell}'\}$ activated $S$, it's easy to see that the activation of $\{u_1,\ldots, u_{\ell}\}$ must activate all of $V(G)$.
\end{proof}

\section{\textsc{Snapshot} problem on Cliques}\label{clique}
\par We briefly discuss here the case of the \textsc{Snapshot} problem restricted on cliques, a special case that unlike others (especially trees) has received little attention in the literature. It's easy to notice that the \textsc{Target Set} problem can be solved efficiently on cliques. When trying to activate the whole clique, and because of the strong symmetry of the graph, the best way we can use our budget of $k$ nodes is on the set of $k$ nodes with the highest threshold. 
\par In the \textsc{Snapshot} problem though, we are interested in activating a specific $S\subseteq V(G)$ and nothing more. When we are under monotone sequential dynamics, the problem is still easy since we get to choose which nodes to best respond at each time period and hence we can just ignore all nodes that are not in the snapshot and then solve the problem on $G[S]$, the subgraph induced by $S$, which is easy to do.
\par However under monotone simultaneous dynamics, we need to be careful to not over-activate, and hence choosing the strongest seed set (i.e., the $k$ nodes of highest thresholds) may not be optimal, as shown next.

\begin{example} Suppose we have a clique of size 10 with thresholds as follows (in increasing order): $t_{u_1}=t_{u_2}=1, t_{u_3}=t_{u_4}=2, t_{u_5}=3, t_{u_6}=4, t_{u_7}=5, t_{u_8}=6, t_{u_9}=7$ and $t_{u_{10}}=8$. Suppose we are given the snapshot $S=\{u_1, u_2,\ldots, u_7\}$ and $k=2$. Then activating the two nodes in $S$ with the highest thresholds, i.e. $u_6$ and $u_7$, will activate four nodes, $u_1, u_2, u_3, u_4$, bringing the total number of activated nodes to 6. After that 
the remainder of $S$ will be activated, but so will $u_8$ causing the snapshot to overshot. If instead we activated $u_5$ and $u_4$, we would be able to activate exactly $S$. The intuition behind this is that we need to keep a balance between the nodes that we need to activate and the nodes that we should not. 
\end{example}

\par We present here some properties that can be used to make the \textsc{Monotone Simultaneous Problem} on cliques simpler. We leave as an open question if these properties can be used to provide provable guarantees on the size of the resulting instance. The proofs are deferred to Appendix \ref{omitted_proofs_clique}

\begin{properti} If there exists a node $u\in S$ such that $t_u\geq |S|$, then $u$ must be in the seed set if $S$ is feasible.\end{properti}
\begin{proof}
It's clear that $u$ could not have been activated by best response and hence it must have been part of the seed set.
\end{proof}

\begin{properti}
If there is a node $u\notin S$ such that $t_u\leq k$, then $S$ is not feasible, unless $|S|=k$. 
\end{properti}
\begin{proof}
If $|S|=k$ the problem is trivial, since we can conclude that $S$ is feasible by having $S$ as the seed set and hence we can focus on the case that $|S|>k$. Then the activation process must evolve for at least a round, and since in that around there will be at $k$ nodes activated any node in $u\in V(G)$ with $t_u\leq k$ must activate as well.
\end{proof}

\begin{properti}
If there are nodes $u\in S$ and $v\notin S$ such that $t_u=t_v$, then $u$ must be part of the seed set if $S$ is feasible.
\end{properti}
\begin{proof}
Because of the strong symmetry of the graph, all nodes with the same threshold that don't belong to the seed set get activated at the same time.
\end{proof}

\begin{properti}
Let $t=min\{t_u | u\notin S\}$. Then we can remove all nodes from $V(G)\setminus S$ that have threshold higher than $t$. 
\end{properti}
\begin{proof}
Between the nodes that are not in the seed set, the ones with lower thresholds get activated weakly earlier from the ones with higher thresholds. Hence it suffices to find a way to activate $S$ that does not activate the nodes with the lowest thresholds outside $S$.
\end{proof}
\begin{properti} Let $t=min\{t_u | u\notin S\}$. If $|S|<t$, then $S$ is feasible if and only if the $k$ highest threshold nodes in $S$ can activate $S$.
\end{properti}
\begin{proof}
Since we are under monotone dynamics, the number of activated nodes gets weakly larger at every time step. Since $|S|<t$, there is no risk of over-activating and hence the problem reduces to the \textsc{Target Set} problem on $G[S]$, the subgraph induced by $S$.
\end{proof}

\section{Conclusions and Open Problems}\label{conclusions}
In this paper we studied the problem of explaining given snapshots of a network diffusion, i.e., finding a small seed set whose activation will cause the activation of the snapshot. Although we presented strong hardness results for all variations we studied, we also presented a variety of structural results that can help better our understanding of this important problem. These structural results could potentially be used as part of heuristics and/or approximation algorithms.
We leave several interesting directions open. One being the complexity of the \textsc{Monotone Snapshot} on cliques. If it's proven to be hard, then the question of polynomial size kernels and approximation algorithms arises, and our results from Section \ref{clique} could potentially be used towards those directions.
Another interesting question is how far can a seed set be from a given snapshot? Can the distance be as large as the diameter of the graph or is it upper bounded by a function of $t_{max}$, the maximum threshold of the graph, and/or the budget $k$.
\subsection*{Acknowledgments}
We thank Nicole Immorlica and Ming-Yang Kao for useful discussions as well as two anonymous referees for their comments.

\bibliography{snapshots_arxiv}{}
\bibliographystyle{plain}
\end{document}